\documentclass[aos,preprint]{imsart}

\RequirePackage[OT1]{fontenc}
\RequirePackage[leqno,cmex10]{amsmath}
\RequirePackage{amsthm}
\RequirePackage[numbers]{natbib}
\RequirePackage[colorlinks,citecolor=blue,urlcolor=blue]{hyperref}

\usepackage{mathrsfs,dsfont}
\usepackage{latexsym}
\usepackage{amssymb}
\usepackage[english]{babel}
\usepackage[latin1]{inputenc}
\usepackage{color}
\usepackage{amsfonts}
\usepackage{bbm}
\usepackage{enumerate}
\RequirePackage{xspace}
\usepackage[ruled, section]{algorithm}
\usepackage{algpseudocode}
\usepackage[paper=letterpaper,dvips,top=1.5in,left=1.5in,right=1.5in,
   foot=1cm,bottom=1in]{geometry}
\usepackage{subfigure}
\usepackage {graphicx}
\usepackage {pgf}
\usepackage{yhmath}
\usepackage[normalem]{ulem}

\usepackage{multido}
\usepackage{pstricks,pst-plot,pstricks-add,pst-math}

\usepackage{mathtools}

\usepackage{enumerate}

\numberwithin{equation}{section}
\numberwithin{figure}{section}
\numberwithin{table}{section}
\theoremstyle{plain}
\newtheorem{thm}{Theorem}[section]
\newtheorem{lem}[thm]{Lemma}

\theoremstyle{definition}

\newcommand{\R}{ \mathds R}
\newcommand{\Z}{ \mathds Z}

\theoremstyle{remark}
\newtheorem{rem}[thm]{Remark}
\newcommand{\ol}[1]{\left\langle #1\right\rangle}

\renewcommand{\le}{\leqslant}
\renewcommand{\ge}{\geqslant}

\newcommand{\eps}{\varepsilon}

\newcommand{\norm}[1]{\left\Vert#1\right\Vert}

\newcommand{\eg}{\emph{e.g.,}\ }

\def\qed{\hfill $\blacksquare$}
\let\ga=\alpha \let\gb=\beta \let\gc=\gamma \let\gd=\delta \let\gee=\epsilon
     \let\gl=\lambda

\let\gC=\Gamma

\newcommand{\cC}{\mathcal{C}}

\newcommand{\cN}{\mathcal{N}}

\newcommand{\V}[1]{\ensuremath{\boldsymbol{#1}}\xspace}

\newcommand{\vzero}{\mathbf{0}}\newcommand{\vone}{\mathbf{1}}

\newcommand{\vA}{\mathbf{A}}\newcommand{\vB}{\mathbf{B}}
\newcommand{\vD}{\mathbf{D}}

\newcommand{\vM}{\mathbf{M}}
\newcommand{\vP}{\mathbf{P}}\newcommand{\vQ}{\mathbf{Q}}
\newcommand{\vR}{\mathbf{R}}
\newcommand{\vU}{\mathbf{U}}\newcommand{\vV}{\mathbf{V}}\newcommand{\vW}{\mathbf{W}}
\newcommand{\vX}{\mathbf{X}}\newcommand{\vZ}{\mathbf{Z}}

\newcommand{\vx}{\mathbf{x}}
\newcommand{\vy}{\mathbf{y}}\newcommand{\vz}{\mathbf{z}}

\newcommand{\dR}{\mathds{R}}

\newcommand{\sA}{\mathscr{A}}

\newcommand{\sS}{\mathscr{S}}

\DeclareMathOperator{\E}{\mathds{E}}
\DeclareMathOperator{\pr}{\mathds{P}}

\newcommand{\sM}{\mathscr{M}}

\def\beq{ \begin{equation} }
 \def\eeq{ \end{equation} }
 \def\beqx{ \begin{equation*} }
 \def\eeqx{ \end{equation*} }
 \def\beqa{\begin{eqnarray}}
 \def\eeqa{\end{eqnarray}}
 \def\beqax{\begin{eqnarray*}}
 \def\eeqax{\end{eqnarray*}}

\newcommand{\Pro}{ \mathbb P}

\newcommand{\gre}{\epsilon}

\newcommand{\GD}{\boldsymbol\Delta}

\usepackage[T3,T1]{fontenc}
\DeclareSymbolFont{tipa}{T3}{cmr}{m}{n}
\DeclareMathAccent{\inv}{\mathalpha}{tipa}{16}

\begin{document}

\begin{frontmatter}
\title{Spectral Clustering for Multiple Sparse Networks: I}
\runtitle{Spectral Clustering for Multiple Sparse Networks}

\begin{aug}
\author{\fnms{Sharmodeep} \snm{Bhattacharyya}\thanksref{m2}\ead[label=e1]{bhattash@science.oregonstate.edu}}
\and
\author{\fnms{Shirshendu} \snm{Chatterjee}\thanksref{m1}\ead[label=e2]{shirshendu@ccny.cuny.edu}}

\runauthor{Bhattacharyya and Chatterjee}

\affiliation{City University of New York\thanksmark{m1}}
\affiliation{Oregon State University\thanksmark{m2}}

\address{Department of Statistics\\
239 Weniger Hall\\
Corvallis, OR, 97331\\
\printead{e1}}
\address{City University of New York, City College\\
New York, NY, 10031
\printead{e2}}

\end{aug}

\begin{abstract} 
\indent Although much of the focus of statistical works on networks has been on static networks, multiple networks are currently becoming more common among network data sets. Usually, a number of network data sets, which share some form of connection between each other are known as multiple or multi-layer networks. We consider the problem of identifying the common community structures for multiple networks. We consider extensions of the spectral clustering methods for the multiple sparse networks, and give theoretical guarantee that the spectral clustering methods produce consistent community detection in case of both multiple stochastic block model and multiple degree-corrected block models. The methods are shown to work under sufficiently mild conditions on the number of multiple networks to detect associative community structures, even if all the individual networks are sparse and most of the individual networks are below community detectability threshold. We reinforce the validity of the theoretical results via simulations too. 
\end{abstract}



\end{frontmatter}

\section{Introduction}
\label{sec_intro}
The analysis of networks has received a lot of attention, not only from statisticians but also from social scientists, mathematicians, physicists and computer scientists. Several statistical methods have been applied to analyze network datasets arising in various disciplines. Examples include networks originating from biosciences such as gene regulation networks \cite{emmert2014gene}, protein protein interaction networks \cite{de2010protein}, structural \cite{rubinov2010complex} and functional networks \cite{friston2011functional} of brain and epidemiological networks \cite{reis2007epidemiological}; networks originating from social media such as Facebook, Twitter and LinkedIn \cite{faloutsos2010online}; citation and collaboration networks \cite{lehmann2003citation}; information and technological networks such as web-based networks, power networks \cite{pagani2013power} and cell-tower networks \cite{isaacman2011identifying}.

Most of the research in the statistics community focuses on developing methods for addressing statistical inference questions based on a single observed network as data. We will refer to such single networks as {\it static networks} in this paper. 
In this paper, we focus on the problem of community detection in networks. The problem of community detection can be considered a sub-problem of vertex clustering problem. In the vertex clustering problem, the goal is grouping the vertices of the graph based on some common properties. In community detection problem, the main goal is grouping the vertices of the graph such that the average number of connections within the group are \emph{either significantly more or less} than the average number of connections between groups. Communities in networks are usually called \emph{associative}, if the average number of connections within communities is \emph{significantly greater} than the average number of connections between communities. 
In this paper, we shall focus on finding \emph{associative community structures}. More rigorous definition of the associative communities will be given later in the paper in section \ref{sec_latent_var_models}. 

Several random graph models has been proposed in the literature, where mathematically rigorous definition of community labels for vertices are given. Examples of random graph models for static networks with community structure include stochastic block models \cite{holland1983stochastic}, degree-corrected block models \cite{karrer2011stochastic} and random dot product models \cite{young2007random}.
A number of methods has also been proposed in the literature for community detection methods (see \cite{fortunato2010community}) for reviews) for static networks. The methods can be broadly classified into two types - \emph{model based approaches}, where, the methods has been developed under the regime of a specific random graph model (e.g.~different likelihood based methods \cite{bickel2009nonparametric}) and model agnostic approaches, where, the methods has been developed without the help of a specific random graph model. (e.g.~modularity based methods \cite{newman2004finding}, spectral clustering methods \cite{rohe2011spectral}, label propagation \cite{gregory2010finding}). We focus on spectral clustering methods for community detection in this paper. 

Since its introduction in \cite{MR0318007}, spectral analysis of various matrices associated to graphs has become one of the most widely used clustering techniques in statistics and machine learning. The advantages of spectral clustering based methods are manifold. Firstly, it is a model agnostic method. So, the spectral clustering methods are not based on any specific random graph model. Secondly, it is highly scalable as the main numerical procedure within it is matrix factorization and a lot of research effort has been employed for scalable implementation of the matrix factorization algorithms in the numerical analysis literature. Thirdly, accuracy of spectral clustering methods in recovering communities has also been shown under various probabilistic models \cite{rohe2011spectral}. In the context of finding clusters in a static unlabeled graph, a number of variants of spectral clustering have been proposed. These methods involve spectral analysis of either the adjacency matrix or some other derived matrix (e.g.~one of the Laplacian matrices) of the graph.  See \cite{shi2000normalized}, \cite{ng2002spectral}, \cite{MR2396807}, \cite{MR2893856}, \cite{lei2014consistency}, \cite{bhattacharyya2014community} \cite{} for some of the research in this regard. Many of these spectral clustering methods have also been theoretically proven to be effective in identifying  communities of static networks, if the networks are generated from some form of exchangeable random graph models \cite{sussman2012consistent}. 

Although much of the focus of statistical works on networks has been on static networks, \emph{multiple networks} are currently becoming more common among network data sets. Usually, a number of network data sets, which share some form of connection between each other are known as \emph{multiple networks}. Various types of multiple networks are becoming common in the literature (see \cite{boccaletti2014structure} for review). Time-evolving networks are one of the common ways of obtaining multiple networks. Other examples include multi-layer networks, multi-dimensional networks, multiplex networks, multi-level networks and hypergraphs, to name a few.
Time-evolving networks are becoming common in many application domains ranging from biological networks (e.g.~genetic \cite{hecker2009gene} or neurological networks \cite{leonardi2013principal}) to social networks \cite{tantipathananandh2007framework}. 

There has also been quite a bit of work on probabilistic models of time-evolving networks. Broadly speaking, there are two main classes of time-evolving network models that have been considered in the literature - (i) network models where both vertex and edge sets change over time (e.g.~preferential attachment models \cite{barabasi1999emergence}) and (ii) network models where the vertex set remains the same, but the edge set changes with time (e.g.~evolving Erd\'{o}s-R\'{e}nyi graph models \cite{crane2015time}). See \cite{goldenberg2010survey} for an early survey on time-evolving network models. In this paper, we will focus on the second kind of time-evolving network models, which we call {\it dynamic network models}. These type pf probabilistic models try to represent time-evolving and multi-layer networks, where networks are represented by a sequence of snapshots of the networks at discrete time steps and the networks share the same vertex set. Thus, the methods proposed and analysis done in this paper can be applied to both time-evolving and multi-layer networks.

Most of the statistical and probabilistic models for dynamic network data sets that appear in the literature  are extensions of random graph models for static networks into dynamic setting.  Examples of such models include extension of latent space models \cite{sarkar2005dynamic}, \cite{sewell2014latent}, extension of  mixed membership block models \cite{ho2011evolving}, extension of random dot-product models \cite{tang2013attribute}, extension of stochastic block models \cite{xu2014dynamic}, \cite{xu2015stochastic}, \cite{matias2015statistical}, \cite{ghasemian2016detectability}, \cite{corneli2016exact}, \cite{zhang2016random}, \cite{pensky2016dynamic}, and extension of Erd\'{o}s-R\'{e}nyi graph models \cite{crane2015time}. Also, some Bayesian models and associated inference procedures have been proposed in the context of dynamic networks  \cite{yang2011detecting}, \cite{durante2014nonparametric}. 

Several approaches have also been put forward in the statistics literature to develop statistical frameworks for inference on dynamic and multi-layer network models. While most of the statistical inference methods developed based on different time-evolving and multi-layer network models are not developed with the goal of community detection, but many of them perform community detection as part of the statistical inference method. So, most of the works like \cite{matias2015statistical}, perform model-based community detection. Although, \cite{chen2017multilayer} have proposed model agnostic algorithms for community detection, the methods do not work when individual networks of the  are sparse. So, works like \cite{han2015consistent} introduce probabilistic models for time-evolving and multi-layer networks with community structure and use approximate likelihood (like profile likelihood) based methods for community detection. Approximate likelihood methods like variational approximation have polynomial time algorithms but lack in theoretical results, where as methods like profile likelihood have theoretical justifications but only have approximate algorithms. 
These approaches limit the scalability of the methods and make the methods very much model dependent.  

Realizing the above limitations of the existing approaches for doing statistical inference on dynamic networks and recognizing  the advantages of using spectral clustering methods (e.g.~scalability and  model agnostic nature) in case of static networks, we propose to use spectral clustering methods for addressing the community detection problem in certain dynamic networks.
 We  also provide theoretical guarantee for the performance of the proposed spectral clustering methods to identify communities in the targeted dynamic network models. The dynamic network models that we use in this paper are similar in spirit to those used in \cite{xu2014dynamic} and \cite{matias2015statistical}.


\subsection{Contribution of our work} 
\label{sec_contribution}
The main contribution of our work are -
\begin{itemize}
\item[(a)] We propose two spectral clustering methods for identifying communities in dynamic or multiple networks. The methods can be used for community detection in single static networks too. The methods are flexible enough to work for both sparse and dense networks.
\item[(b)] We also prove analytically that, under very mild parametric conditions, the proposed spectral clustering methods perform consistently to identify communities for networks generated from dynamic block models and dynamic degree-corrected block models. We show that in the above dynamic network settings, spectral clustering can recover underlying common community structure even if the individual networks are extremely   sparse (e.g.~have constant average degree). 
\end{itemize}

In this paper, we only consider the case, where the community membership does not change with time. However, the methods will still work if the community memberships change by a vanishing fraction at each time point. Some possible extensions of our work will include considering the cases when cluster memberships change significantly with time and the dynamic behavior of the network is more general.

\subsection{Structure of the paper} 
\label{sec_organization}
The remainder of the paper is organized as follows. 
In section \ref{sec_method}, we describe the spectral clustering methods in the dynamic clustering setting. In section \ref{sec_theory}, we state the theoretical results regarding the performance of the proposed spectral clustering methods. We also give the proofs of consistency. 

\section{Community Detection Algorithms}
\label{sec_method}
We consider a sequence of random unlabeled graphs $G_n^{(t)}, t=1, \ldots, T,$ on the vertex set $V_n=\{v_1, v_2, \ldots, v_n\}$ having $n$ vertices as the observed data. Note that the vertex sets $V(G_n^{(t)})$ of $G_n^{(t)}$ don't change with $t$, and $|V(G_n^{(t)})| = n$ for all $t$. The edge set $E(G^{(t)}_n)$  may be different for each $t$. 
We shall consider undirected, unweighted graphs only in this paper. However the conclusions of the paper can be extended for fixed weighted graphs in a quite straightforward way. 

As usual, we suppose that the network corresponding to the graph $G^{(t)}$ is represented by an adjacency matrix $\vA^{(t)}_{n\times n}$ whose elements are $\vA^{(t)}_{ij}\in\{0,1\}$. $\vA^{(t)}_{ij} = 1$, if node $v_i$ links to node $v_j$ at time $t$, and $\vA^{(t)}_{ij} = 0$ otherwise. 
We will refer to the network model at a specific time $t$ as the \emph{static model at time $t$}.

\subsection{Notations}
\label{sec_notation}
Let $[n] := \{1, 2, \ldots, n\}$ for any positive integer $n$, $\sM_{m,n}$ be the set of all $m\times n$ matrices which have exactly one 1 and $n-1$ 0's in each row. $\R^{m\times n}$ denotes the set of all $m\times n$ real matrices. $||\cdot||_2$ is used to denote Euclidean $\ell_2$-norm for vectors in $\R^{m\times 1}$. $||\cdot||$ is the spectral norm on $\R^{m\times n}$. $||\cdot||_F$ is the Frobenius norm on $\R^{m\times n}$, namely $||M||_F := \sqrt{trace(M^T M)}$. $\vone_{m} \in \R^{m\times 1}$ consists of all 1's, $\mathbf 1_A$ denotes the indicator function of the event $A$. 
For $\vA\in\dR^{n\times n}$, $\cC(\vA)$ and $\cN(\vA)$ denote its column space and null space of $\vA$ respectively, and $\gl_1(\vA), \gl_1^+(\vA)$ denote the smallest and smallest positive eigenvalues of $\vA$. 
If $\vA\in\R^{m\times n}$,  $I\subset [m]$ and $j\in [n]$, then $\vA_{I,j}$ (resp.~$\vA_{I,*}$)  denotes the submatrix  of $\vA$ corresponding to row index set $I$ and column index $j$ (resp.~index set $[n]$).

\subsection{Dynamic Stochastic Block Model}
The first model that we consider is a version of the time-evolving SBM. We will refer to this model as {\it dynamic stochastic block model} (DSBM) in the paper. DSBM for $K$ communities $\cC_1, \ldots, \cC_K$ can be described in terms of two parameters: (i) the membership vector $\V{z}=(z_1, \ldots, z_n)$, where each $z_i \in \{1, \ldots, K\}$, and (ii)  the $K\times K$ connectivity probability matrices $\vB:=\left(\vB^{(t)}: 1\le t\le T\right)$. The DSBM having parameters $(\V{z}, \V{\pi}, \V{B})$ is given by
\begin{eqnarray}
	\label{eq_sbm0}
	\vz_1, \ldots, \vz_n & \stackrel{iid}{\sim} & \mbox{Mult}(1;(\pi_1,\ldots,\pi_K)),\\
	\label{eq_sbm1}
	\Pro\left(A^{(t)}_{ij} = 1\right| \vz_i, \vz_j) & = & B^{(t)}_{\vz_i \vz_j}. 
\end{eqnarray}

Suppose $\vZ \in \sM_{n,K}$ denotes the actual membership matrix. $\vZ$ is unknown and we wish to estimate it. If for $i \in [n]$ the corresponding community index is $\vz_i \in [K]$, then clearly  
\[ \vZ_{ij} = \mathbf 1_{\{\vz_i=j\}},\]
In a DSBM$(\vz, \V{\pi}, \vB)$, independent edge formation is assumed given the edge probability matrices $\vP^{(t)}:=(P^{(t)}_{ij})_{i,j\in[n]}$. 
So, for $i, j \in [n]$ with $i \ne j$ and for $t \in [T]$
\beq \label{A^t bmdef}
A^{(t)}_{i,j} \sim Bernoulli(P^{(t)}_{i,j}), \text{ where }
\vP^{(t)} := \vZ \vB^{(t)}\vZ^T.\eeq

\subsection{Dynamic Degree Corrected Block Model}
The other model that we consider in this paper is an extension of the degree corrected  block model (DCBM) to the dynamic setting. The dynamic degree-corrected block model (DDCBM) for $K$ communities $\cC_1, \ldots, \cC_K$ can be described in terms of three parameters: (i) the membership vector $\V{z}=(z_1, \ldots, z_n)$, where each $z_i \in \{1, \ldots, K\}$,  (ii)  the $K\times K$ connectivity probability matrices $\vB:=\left(\vB^{(t)}: 1\le t\le T\right)$, and (iii)
a given set of \emph{degree parameters} $\V{\psi} = (\psi_1, \ldots, \psi_n)$. The DDCBM having parameters $(\V{z}, \V{\pi}, \V{B}, \V{\psi})$ is given by
\begin{eqnarray}
	\label{eq_dcbm0}
	\vz_1, \ldots, \vz_n & \stackrel{iid}{\sim} & \mbox{Mult}(1;(\pi_1,\ldots,\pi_K)),\\
	\label{eq_dcbm1}
	\Pro\left(A^{(t)}_{ij} = 1\right| \vz_i, \vz_j) & = & \psi_i\psi_jB^{(t)}_{\vz_i \vz_j}. 
\end{eqnarray}

The inclusion of $\V{\psi}$ involves the obvious issue of identifiability. In order to avoid this issue  we assume that 
\begin{align}
	\label{eq_dcbm_id}  
	\max_{i\in\cC_k} \psi_i=1 \text{ for all } k\in\{1, 2, \ldots, K\}.   
\end{align}

The identifiability of the models described in \eqref{eq_sbm1} and \eqref{eq_dcbm1} have been proven by Matias and Miele (2016) \cite{matias2015statistical}, so we will not elaborate on that. 
For the dynamic network models described in \eqref{eq_sbm1} and  \eqref{eq_dcbm1}, we shall try to estimate the underlying latent variables $\V{z}$  using spectral clustering methods.

In an DDCBM$(\vz, \V{\psi}, \V{\pi}, \vB)$ also independent edge formation is assumed given the edge probability matrices $\tilde\vP^{(t)}$. Here also, for $i, j \in [n]$ with $i \ne j$ and for $t \in [T]$
\beq \label{A^t def}
A^{(t)}_{i,j} \sim Bernoulli(\tilde P^{(t)}_{i,j}), \text{ where }
\tilde\vP^{(t)}  := Diag(\V{\psi})\vZ \vB^{(t)}\vZ^TDiag(\V{\psi}).
\eeq

\subsection{Spectral Clustering  for Sum of Adjacency Matrices}
\label{sec_algo_1}
We apply the spectral clustering method using \textbf{sum of the adjacency matrices} 
\[ \vA_0=\sum_{t=1}^T \vA^{(t)}.\]
Define $\bar{d} = \frac{1}{nT}\mathbf{1}_n^T\vA_0\mathbf{1}_n$ to be the average degree of a node. 
For $\vA_0$, let $n'$ be the number of rows and $1\le k_1<k_2<\cdots<k_{n'}\le n$ be the row indices having row sum at most $e(T\bar{d})^{5/4}$. 
Let $\vA\in\R^{n'\times n'}$ be the submatrix of $\vA_0$ such that $A_{i,j} := (A_0)_{k_i,k_j}$ for $i, j\in [n']$. Next, we obtain the leading $K$ eigenvectors of $\vA$ corresponding to its largest absolute eigenvalues. Suppose $\hat \vU\in \R^{n'\times K}$ contains those eigenvectors as columns. Then, we use an $(1+\gre)$ approximate $K$-means clustering algorithm on the row vectors of $\hat\vU$ to obtain $\hat\vZ \in \sM_{n',K}$ and $\hat \vX \in \R^{K \times K}$ such that
\begin{align} \label{eq:kmeans}
	||\hat\vZ\hat\vX - \hat \vU||_F^2 \le (1+\gre) \min_{\V{\gC} \in \sM_{n'\times K}, \vX \in \R^{K\times K}}  ||\V{\gC} \vX - \hat\vU||_F^2.
\end{align}
Finally, $\hat\vZ$ is extended to $\hat\vZ_0\in\sM_{n,K}$ by taking $(\hat\vZ_0)_{k_j,*}:=\hat\vZ_{j,*}, j\in[n'],$  and filling in the remaining rows arbitrarily. 
\[ (\hat\vZ_0)_{i,*} := \begin{cases} \hat\vZ_{j,*} & \text{ if  $i=k_j$ for some $j\in[n']$} \\ (e^K_1)^T & \text{otherwise}\end{cases}\]
$\hat\vZ_0$ is the estimate of $\vZ$ from this method. The reason for using an $(1+\gre)$ approximate $K$-means clustering algorithm is that the $K$-means clustering is originally an NP-hard problem with only $(1+\gre)$-approximate solutions available. 

\vspace{0.2in}
\framebox[\textwidth]
{\centering\parbox{.95\textwidth}{
		\textbf{Algorithm 1:} Spectral Clustering of the Sum of  the Adjacency Matrices \\
		\textbf{Input:} Adjacency matrices $\vA^{(1)}, \vA^{(2)}, \ldots, \vA^{(T)}$; number of communities $K$; approximation parameter $\gre$.\\
		\textbf{Output:} Membership matrix $\hat\vZ_0$. \\
		\begin{enumerate}
			\item Obtain the \textbf{sum of the adjacency matrices},
			$\vA_0=\sum_{t=1}^T \vA^{(t)}$.
			\item Get $\bar d :=\frac{1}{nT}\vone_n^T \vA_0\vone_n$. Let $n'$ be the number of rows (having indices $1\le k_1<k_2<\cdots <k_{n'}\le n$) of $\vA_0$ with row sum at most $e(T\bar d)^{5/4}$.
			\item Let $\vA\in\R^{n'\times n'}$ be the submatrix of $\vA_0$: $A_{i,j}=(A_0)_{k_i,k_j}$ for $i, j\in[n']$. 
			\item Obtain $\hat\vU\in\R^{n'\times K}$ consisting of the leading $K$ eigenvectors of $\vA$ corresponding to its largest absolute eigenvalues. 
			\item Use $(1+\gre)$ approximate $K$-means clustering algorithm on the row vectors of $\hat\vU$ to obtain $\hat{\vZ} \in \sM_{n',K}$ and $\hat \vX \in \R^{K \times K}$ satisfying \eqref{eq:kmeans}.
			\item Extend $\hat\vZ$ to obtain $\hat\vZ_0\in\sM_{n,K}$ as follows. $(\hat\vZ_0)_{i,*} = \hat\vZ_{j,*}$ (resp.~$(1,0,\ldots, 0)$) for $i=k_j$ (resp.~$i\notin\{k_1, \ldots, k_{n'}\}$). 
			\item $\hat\vZ_0$ is the estimate of $\vZ$.
		\end{enumerate}
}}\\

\subsection{Spherical Spectral Clustering Algorithm for Sum of Adjacency Matrices} 
\label{sec_algo_2}
We apply the Spherical Spectral Clustering method using \textbf{sum of the adjacency matrices} 
\[ 
\vA_0=\sum_{t=1}^T \vA^{(t)}.\]
The spherical spectral clustering method is modification of the method described in Section \ref{sec_algo_1} based on the spherical spectral clustering algorithm, which was proposed in Jin (2015) \cite{jin2015fast} and used in Lei and Rinaldo (2015) \cite{lei2014consistency}. We will use the norm $\norm{\cdot}_{2,1}$ on $\R^{m\times n}$ defined by $\norm{\vM}_{2,1}:=\sum_{i=1}^m \norm{\vM_{i,*}}_2$. 

In Algorithm 2, we describe the spherical spectral clustering method using the truncated sum of the adjacency matrices $\vA$. 
We obtain $\hat\vU$ as earlier. Recall that $\hat\vU \in\R^{n'\times K}$ contains the leading $K$ eigenvectors (corresponding to the largest absolute eigenvalues) of $\vA$ as columns.  Let $n''$ be the number of nonzero rows (having indices $1\le l_1 < l_2 <\cdots <l_{n''}\le n'$) of $\hat\vU$. Let $\hat\vU^+ \in \R^{n''\times K}$ consist of the normalized nonzero rows of $\hat\vU$, i.e. $\hat\vU^+_{i,*}=\hat\vU_{l_i,*}/\norm{\hat\vU_{l_i,*}}_2$ for $i\in [n'']$.  Then we use $(1+\gee)$ approximate $K$-median clustering algorithm on the rows of $\hat\vU^+$ to obtain $\check\vZ^+ \in \sM_{n'',K}$ and $\check \vX \in \R^{K \times K}$ such that
\beq \label{k median}
\norm{\check\vZ^+\check \vX - \hat\vU^+}_{2,1} \le (1+\gee) \min_{\V{\gC} \in \sM_{n''\times K}, \vX \in \R^{K\times K}}  \norm{\V{\gC} \vX - \hat\vU^+}_{2,1}.\eeq
Finally, $\check\vZ^+$ is extended to $\check\vZ\in\sM_{n',K}$, and then 
$\check\vZ$ is extended to $\check\vZ_0\in\sM_{n,K}$ by taking $\check\vZ_{l_j,*}:=\check\vZ^+_{j,*}, j\in [n''],$ and $(\check\vZ_0)_{k_j,*}:=\check\vZ_{j,*}, j\in [n'],$  and filling in the remaining rows arbitrarily.  
\[ \check\vZ_{i,*} := \begin{cases} \check\vZ^+_{j,*} & \text{ if  $i=l_j$} \\ (e^K_1)^T & \text{if $i\notin\{l_1, \ldots, l_{n''}\}$}\end{cases},
(\check\vZ_0)_{i,*} := \begin{cases} \check\vZ_{j,*} & \text{ if  $i=k_j$} \\ (e^K_1)^T & \text{if $i\notin\{k_1, \ldots, k_{n'}\}$}\end{cases}
\]
$\check\vZ_0$ is the estimate of $\vZ$ from this method. As in the previous case, the reason for using an $(1+\gre)$ approximate $K$-medians clustering algorithm is that the $K$-medians clustering is originally an NP-hard problem with only $(1+\gre)$-approximate solutions available. 

\vspace{0.2in}

\framebox[\textwidth]
{\centering\parbox{.95\textwidth}{
		\textbf{Algorithm 2:} Spherical Spectral Clustering of the Sum of the Adjacency Matrices \\
		\textbf{Input:} Adjacency matrices $\vA^{(1)}, \vA^{(2)}, \ldots, \vA^{(T)}$; number of communities $K$; approximation parameter $\gre$. \\
		\textbf{Output:} Membership matrix $\check\vZ$. \\
		\begin{enumerate}
			\item Obtain the \textbf{sum of the adjacency matrices} $\vA_0=\sum_{t=1}^T \vA^{(t)}$.
			\item Get $\bar d :=\frac{1}{nT}\vone_n^T \vA_0\vone_n$. Let $n'$ be the number of rows (having indices $1\le k_1<k_2<\cdots <k_{n'}\le n$) of $\vA_0$ having row sum at most $e(T\bar d)^{5/4}$.
			\item Let $\vA\in\R^{n'\times n'}$ be the submatrix of $\vA_0$: $A_{i,j}:=(A_0)_{k_i,k_j}$ for $i,j\in[n']$. 
			\item Obtain $\hat\vU\in\R^{n'\times K}$ consisting of the leading $K$ eigenvectors of $\vA$ corresponding to its largest absolute eigenvalues. 
			\item Let $n''$ be the number of nonzero rows (having indices $1\le l_1<l_2<\cdots <l_{n''}\le n'$) of $\hat\vU$. Obtain $\hat\vU^+ \in \R^{n''\times K}$ consisting of normalized nonzero rows of $\hat\vU$, i.e.~$\hat\vU^+_{i,*}=\hat\vU_{l_i,*}/\norm{\hat\vU_{l_i,*}}_2$ for $i\in[n'']$. 
			\item Use $(1+\gre)$ approximate $K$-median clustering algorithm on the row vectors of $\hat\vU^+$ to obtain $\check\vZ^+ \in \sM_{n'',K}$ and $\check \vX \in \R^{K \times K}$ obeying \eqref{k median}.
			\item Extend $\check\vZ^+$ to obtain $\check\vZ\in\sM_{n',K}$ as follows. $\check\vZ_{j,*} = \check\vZ^+_{i,*}$ (resp.~$(1,0,\ldots, 0)$) for $j=l_i$ (resp.~$j\notin\{l_1, \ldots, l_{n''}\}$). 
			\item Extend $\check\vZ$ to obtain $\check\vZ_0\in\sM_{n,K}$ as follows. $(\check\vZ_0)_{j,*} = \check\vZ_{i,*}$ (resp.~$(1, 0,\ldots, 0)$) for $j=k_i$ (resp.~$j\notin\{k_1, \ldots, k_{n'}\}$). 
			\item $\check\vZ_0$ is the estimate of $\vZ$.
		\end{enumerate}
}}\\

\subsubsection{Selection of $K$}
\label{sec_k_sel}
In this paper, we do not consider the problem of selection of number of communities $K$. We assume that the number of communities, $K$, is given for the Algorithms 1-4. However, $K$ can also be estimated using the scree plot of the absolute eigenvalues of the matrices $\V{A}$ and $\ol{\V{A}^{[2]}}$. We can use the thresholds proposed in \cite{chatterjee2015matrix} for choosing the  number of communities.

\subsubsection{Parameter Estimation}
\label{sec_par_est}
The basic goal of community detection is to infer the node labels, or equivalently the membership matrix $\V{Z}$, from the data. Although we do not explicitly consider the estimation of the parameters $\pi$ and $\vB$, they can be estimated using an estimate  $\hat{\V{Z}}$ of $\V{Z}$ as follows.
\begin{align}
\label{eq_propnode}
\hat{\pi}_{a}  & := & \frac{1}{n}\sum_{i=1}^n \mathbf{1}\left(\hat{\V{Z}}_i = e_a\right),\ \ \ \ \ \ \ \ \ \ \ \ \ \ \ \ \ \ \ \ \ \ \ 1\leq a\leq K, \\
\label{eq_meanlink}
\hat{B}^{(t)}_{ab} & := & \frac{1}{O_{ab}}\sum_{i=1}^n\sum_{j=1}^n A^{(t)}_{ij}\mathbf{1}\left(\hat{\V{Z}}_i = e_a, \hat{\V{Z}}_j = e_ b\right),\ \ \ 1\leq a,b\leq K,
\end{align}
where
\begin{eqnarray*}
O_{ab}  := & \left\{
	\begin{array}{ll}
	\hat n_a\hat n_b, & 1\leq a\neq b\leq K \\
	\hat n_a(\hat n_a-1), & 1\leq a\leq K, a=b
\end{array}
\right. , &  \hat n_{a}  :=  \sum_{i=1}^n \mathbf{1}\left(\hat{\V{Z}}_i = e_a\right),\ a\in[K] 
\end{eqnarray*}
and $(e_a)_{K\times 1}$ denotes the unit vector with $1$ at $a^{\mbox{th}}$ position ($a\in [K]$).

\section{Theoretical Justification}
\label{sec_theory}

\subsection{Consistency of  Spectral Clustering label $\hat\vZ_0$ under DSBM}
\label{sec_proof_hatZ_supp}
In order to state the result about the consistency of $\hat\vZ_0$ for networks generated from \textbf{DSBM}, we need to assume the following condition on the sum of connection probability matrices - 
\begin{align} \label{eq_sum_ass}
	\sum_{t\in[T]} \vB^{(t)} \text{ must be nonsingular.}
\end{align}

\begin{thm} \label{ConsSum1}
	Let $\{\vA^{(t)}\}_{t=1}^T$ be the adjacency matrices of the networks generated from the DSBM, where
	\begin{itemize}
		\item $T\ge 1$ is the number of networks
		\item $n$ is the number of nodes, $K$ is the number of communities and $n\ge 2K$
		\item $\{\vB^{(t)}\}_{t=1}^T$ satisfy assumption  \eqref{eq_sum_ass}
		\item  $\ga=\ga(n,T):=\max_{a,b \in[K], t\in [T]} B^{(t)}_{a,b}$ is the maximum connection probability of an edge at any time, and $\gl=\gl(n,T)>0$ is such that $\gl\ga$ is the smallest eigenvalue of  $(\vB^{(t)}, t\in[T])$.
		\item $n_{\text{min}}$ is the size of the smallest community.
	\end{itemize}
	For any $\gee>0$ and $c\in(0,1)$, 
	there are constants $C_1=C_1(\gee,c), C_2=C_2(c)>0$ such that if $(f_a, a\in[K])$ denotes the proportion of   nodes having community label $a$, which are misclassified in Algorithm 1 and if $Tn\ga\ge C_2(K/\gl)^5$, then
	\[ \sum_{a\in[K]} f_a \le \left(\frac{n_{\text{min}}}{n}\right)^{-1} e^{-(1-c)Tn\ga} +  C\left(\frac{n_{\text{min}}}{n}-e^{-(1-c)Tn\ga}\right)^{-2}K\gl^{-2}(Tn\ga)^{-1/2} \]
	with probability at least $1-5\exp(-c\min\{Tn\ga\gl,\log(n)\})$.
	
	Therefore, in the special case,  when 
	(i)  $\gl>0$ and \; (ii) the community sizes are balanced, i.e.~$n_{\text{max}}/n_{\text{min}} = O(1)$, then consistency holds for $\hat\vZ_0$ with probability $1-o(1)$ if $Tn\ga\gl\to\infty$.
\end{thm}

\begin{rem}
	The condition ``$Tn\ga\gl\to\infty$" is necessary and sufficient in order to have a consistent estimator of $\vZ$.
	Theorem \ref{ConsSum1} proves the sufficiency. The  necessity of the condition follows from the work of Zhang and Zhou \cite{ZZ16}. Consider a SBM (so $T=1$), where (i) there are two communities having equal size $n$ and  (ii) the within (resp.~between)  community connection probability is $a/n$ (resp.~$b/n$) for some constants $a>b>0$. In this case $Tn\ga\gl=a-b$ which is bounded. In this case, (see\cite{ZZ16}) that there is a constant $c>0$  such that if 
	\[ \frac{(a-b)^2}{a+b} <c\log\frac 1\gc\]
	for some $\gc$, 
	then one cannot recover a partition (in expectation) having proportion of misclassification $<\gc$, regardless the algorithm.
\end{rem}
\begin{rem}
	If we use $(T\bar d)^{1+\gd}$ instead of  $(T\bar d)^{5/4}$ in Algorithm 1, then the bound for $\sum_{a\in[K]}f_a$ will involve $(Tn\ga)^{-(1-2\gd)}$ instead of $(Tn\ga)^{-1/2}$ and we will need $Tn\ga \ge C_2(K/\gl)^{1/\gd+1}$ instead of $(K/\gl)^5$.
\end{rem}

\begin{rem} \label{CompareWithOneNetwork}
	Despite using $\vA$ for spectral clustering, if one uses just one graph for spectral
	clustering and discards all the remaining observations of $\{\vA^{(t)}\}$, then  (assming all $B^{(t)}$ are equal and associative)
	the former algorithm outperforms the later one (with respect to the fraction of nodes mis-clustered) 
	by a factor 
	\[ \begin{cases}
	1/T & \text{ if } d\ge \log(n) \\
	1/T^3 & \text{ if } Td\le \log(n) \\
	(d/\log(n))^2/T & \text{ if } d\le \log(n) \le Td\end{cases}.\]
\end{rem}

\subsubsection{Proof of Theorem \ref{ConsSum1}}
Without loss of generality we can assume that $k_i=i$ for all $i\in[n']$. Define $n_i':=\vone_{n'}^T\vZ_{[n'],i}$ for $i\in[K]$ and
\begin{align}
	\GD := Diag\left(\sqrt{n_1'}, \sqrt{n_2'}, \ldots, \sqrt{n_K'}\right),  \quad \vB  := \sum_{t=1}^T \GD \vB^{(t)}\GD,\notag \\
	\vP  := \sum_{t=1}^T \vZ_{[n'],*} \vB^{(t)}\vZ_{[n'],*}^T = \vZ_{[n'],*}\GD^{-1}\vB\GD^{-1}\vZ_{[n'],*}^T. 
	\label{Pdef}
\end{align}
It is easy to see that if $\vR\vD\vR^T$  is the spectral decomposition of $\vB$, then $\vU:=\vZ_{[n'],*}\GD^{-1}\vR$ consists of the leading $K$ eigenvectors of $\vP$.

The proof of Theorem \ref{ConsSum1}  relies on the estimates provided in Lemma \ref{misclassify bd} and Theorem \ref{normbd}.

\begin{lem} \label{misclassify bd}
	For the estimator $\hat\vZ_0$ of $\vZ$ (as described in Algorithm 1),  $\sum_{a\in[K]} f_a \le (n-n')/n_{min}+32K(4+2\eps)\gc_n^{-2}||\vA-\vP||^2$, where $\vA$ is  described in Algorithm 1, $\vP$ is defined in \eqref{Pdef} and $\gc_n$ denotes the smallest nonzero singular value of $\vP$.
\end{lem}

\begin{thm} \label{normbd}
	Let $\vA$ be the matrix described in Algorithm 1 and $\vP$ be as in \eqref{Pdef}. For any constant $c\in(0,1)$, there are constants $C, C'>0$ (depending on $c$) such that if $Tn\ga \ge C'(K/\gl)^5$ and
	$\sA:=\{||\vA - \vP|| \le C(Tn\ga)^{3/4}\}\cap\{(T\bar d)^{5/4} \ge  Tn\ga\}$, then $\pr(\sA)\ge 1-4\exp[-c\min\{\log(n), Tn\ga\gl\}]$.
\end{thm}

The proof of Lemma \ref{misclassify bd} uses some of the known techniques,  but involves some additional technical details. We present the proof  in Section \ref{misclassify bd proof}.

\begin{rem}
	Theorem \ref{normbd} cannot be proved only using the  conventional matrix concentration inequalities, \eg the matrix Bernstein inequality, which would provides suboptimal bound for the spectral norm (with an extra $\log(n)$ factor). 
\end{rem}

\begin{rem}
	There are some methods (in case of static networks) available in the literature for bounding the spectral norm of centered adjacency matrix, but when these methods are applied on the sum-adjacency matrix, they produce suboptimal bounds. For example,
	Lu \& Peng (2012) use a path counting technique
	in random matrix theory, but their method would require the condition
	$Tn\ga_n \ge c(\log(n))^4$ in order to obtain a similar bound for the spectral norm. In \cite{lei2014consistency} the authors use the Bernstein inequality and a combinatorial argument
	to bound the spectral norm of the entire adjacency matrix (in the static network case), but they need the maximal expected degree to be $\ge c\log(n)$. So if we adopt that method in our setting, we would need the condition $Tn\ga  \ge c\log(n)$.
\end{rem}

The proof of Theorem \ref{normbd} involves intricate technical details, as it uses some large deviation estimates and combinatorial arguments. Our proof is partially based on the techniques used in  \cite{lei2014consistency} (originally developed by
Feige \& Ofek (2005) for bounding the second largest eigenvalue of an Erd\'{o}s-R\'{e}nyi random graph
with edge probability $\ga_n$). The details are provided in Section \ref{normbd proof}. 

\begin{proof}[Proof of Theorem \ref{ConsSum1}]
	First we will bound $n'$.  Note that for any node $i\in[n]$, $\sum_{j\in[n],t\in[T]}A^{(t)}_{i,j}$ is stochastically dominated by $X\sim Binomial(Tn,\ga)$. So using Lemma \ref{BinLDP}  and the properties of the event $\sA$ described in Theorem \ref{normbd}.
	\[\E(n-n';\sA)
	=\sum_{i\in[n]} \pr\left(\sum_{j\in[n], t\in[T]}A^{(t)}_{i,j}>e(T\bar d)^{5/4}; \sA\right) \le n \pr(X\ge eTn\ga) \le ne^{-Tn\ga}.\]
	Using the above estimate and applying Markov inequality,  
	\begin{align}   
		\text{if } & \sA':=\sA\cap\{n'\ge n(1-e^{-(1-c)Tn\ga})\},  \text{ then }\notag \\
		&
		\pr(\sA'^c)
		\le \pr(\sA^c) + \frac{\E(n-n';\sA) }{ne^{-(1-c)Tn\ga}} \le 5\exp(-c\min\{\log n, Tn\ga\gl\}).
		\label{sA^c bd}
	\end{align}
	Using Lemma \ref{misclassify bd} and Theorem \ref{normbd},  there is a constant $C$ such that
	\beq \label{f_a bd}
	\sum_{a\in[K]}f_a \le e^{-(1-c)Tn\ga}\frac{n}{n_{min}}+CK\gc_n^{-2} (Tn\ga)^{3/2} \text{ on the event } \sA'. \eeq
	In order to bound $\gc_n$,
	note that
	\[ \gamma_n = \min_{\vx \in \R^K: \vx\ne \vzero} \frac{\vx^T\vZ_{[n'],*}^T\vP\vZ_{[n'],*}\vx}{\vx^T\vZ_{[n'],*}^T\vZ_{[n'],*}\vx}  \ge \sum_{t\in[T]} \min_{\vx \in \R^K: \vx\ne 0} \frac{\vx^T\GD^2B^{(t)}\GD^2\vx}{\vx^T\GD^2\vx} .\]
	In the definition of $\gamma_n$ we consider only those vectors which belong to $\cC(\vZ_{[n'],*})$, since $\gamma_n$ is the smallest positive eigenvalue. Writing $\vy=\GD \vx$ and $\vz=\GD \vy$, the above is 
	\[ \ge \sum_{t\in[T]} \min_{\vz \in \R^K: \vz\ne 0} \frac{\vz^T\vB^{(t)}\vz}{\vz^T\vz} \min_{\vy \in \R^K: \vy\ne 0} \frac{\vy^T\GD^2\vy}{\vy^T\vy}  \ge T \ga\lambda n'_{\text{min}}.\]
	Plugging this bound for $\gamma_n$ into \eqref{f_a bd}, we get the desired result. 
\end{proof}

\subsubsection{Large deviation estimates}
The following large deviation estimates will be necessary for our argument.

\begin{lem}[\cite{D07}] \label{BinLDP}
	If $X\sim Binomial(N,p)$,  then 
	\[ 
	\pr(X\ge xNp)\le e^{-\gc(x)Np} \quad \text{ and } \quad 
	\pr(X\le yNp)\le e^{-\gc(y)Np} \]
	for all $y\le 1\le x$ where $\gc(x):=x\log(x)-x+1$ is a nonnegative convex function having unique minima at $x=1$.
\end{lem}

\begin{lem} \label{Algo1 A bd}
	For any $c\in(0,1)$, there are constants $C_1(c), C_2(c)>0$ such that if $n\ge 3K, Tn\ga\ge C_2 (K/\gl)^5$ and
	\[\sA_1 := \left\{\frac{1}{C_1}T\bar d \le Tn\ga \le (T\bar d)^{5/4}\right\}, \text{ then } \pr(\sA_1) \ge 1-2e^{-cTn\ga\gl}.\]
\end{lem}

\begin{proof}[Proof of Lemma \ref{Algo1 A bd}]
	Note that $\vone_n^T\vA_0\vone_n$ is stochastically dominated by $2Y$, where $Y\sim Binomial(T{n\choose 2}, \ga)$. so, if $C_1>1$ satisfies $\gc(C_1)\ge c$, then  using Lemma \ref{BinLDP}
	\[
	\pr(T\bar d\ge C_1Tn\ga)
	\le  \pr(Y\ge C_1n^2T\ga/2) \le \exp[-cTn(n-1)\ga/2].
	\]
	On the other hand, $\vone_n^T\vA_0\vone_n$ stochastically dominates $2Z$, where $Z\sim Binomial(Tn(n-K)/(2K), \ga\gl)$, because $B^{(t)}_{a,a} \ge \gl\ga$ for all $a\in[K]$ by the definition of $\gl$ and $\sum_{a\in[K]}{n_a\choose 2} \ge n(n-K)/(2K)$ using Cauchy-Schwartz inequality. Therefore, using Lemma\ref{BinLDP}
	\begin{align*}
		\pr(T\bar d\le (Tn\ga)^{4/5})
		&\le \pr(Z\le \frac n2(Tn\ga)^{4/5}) \\
		& \le \pr(Z\le zT\ga\gl\frac n2(\frac nK-1)) 
		\le \exp[-\gc(z)Tn\ga\gl], 
	\end{align*}
	where $z=(Tn\ga)^{-1/5}\frac 32\frac K\gl$. We choose $C_2$ so that $Tn\ga\ge C_2(K/\gl)^5$ implies $\gc(z)\ge c$. So the upper bound in the last display is at most $e^{-cTn\ga\gl}$. Combining the two estimates we get the result.
\end{proof}

\begin{lem} \label{max degreebd}
	For any $c>0$ there exists a constant $c_1(c)>1$ and an event $\sA_2$ satisfying $\pr(\sA_2) \ge 1-n^{-c}$ such that the following holds on $\sA_2$. For any two subsets $I,J\subset [n]$ satisfying $|I| \le |J| \le n/e$ and $e(I,J):=\sum_{i\in I, j\in J}(A_0)_{i,j}$, 
	\begin{align*}
		\text{ either } & e(I, J) \le e^{4.4}|I|\cdot |J| T\ga \\
		\text{ or } & e(I, J)\log\frac{e(I, J)}{|I|\cdot |J| T\ga} \le c_1  |J|\log\frac{n}{|J|}.
	\end{align*}
\end{lem} 

\begin{proof} Using the fact $P^{(t)}_{i,j} \le \ga$ for all $i, j, t$ it is easy to see that
	$e(I, J)$ is  stochastically dominated by $X+2Y$, where $X$ and $Y$ are independent and 
	\[ X \sim Binomial(T [|I|\cdot|J\setminus I|+|I\cap J|\cdot|I\setminus J|], \ga), \quad Y\sim Binomial(T|I\cap J|(|I\cap J|-1)/2, \ga).\] 
	Using this observation and Markov inequality we see that for any $r>1$ and $\theta>0$,
	\beqax
	&& \pr\left(e(I,J) \ge \frac rn |I| \cdot |J|Tn\ga\right)\\
	& \le & \pr\left(X+2Y \ge \frac rn |I| \cdot |J|Tn\ga\right) 
	\le   e^{-\theta r |I| \cdot |J|T\ga}\E e^{\theta(X+2Y)}  \\
	& = & e^{-\theta r |I| \cdot |J|T\ga} \left(\ga e^\theta + 1 -\ga\right)^{T [|I|\cdot|J\setminus I|+|I\cap J|\cdot|I\setminus J|]} \left(\ga e^{2\theta} + 1 -\ga \right)^{T|I\cap J|(|I\cap J|-1)/2} \\
	& \le &  \exp\left(-\frac 1n Tn\ga|I|\cdot |J|\left[r\theta - (e^\theta-1)(|J\setminus I|/|J|+\frac{|I\setminus J| \cdot|I\cap J|}{|I|\cdot|J|}) - (e^{2\theta}-1)\frac{|I\cap J|(|I\cap J|-1)}{2|I|\cdot |J|}\right]\right).
	\eeqax
	The last inequality follows because $1+x \le e^x$ for all $x \in \dR$. Putting $\theta=\frac 12 \log r$ in the above inequality we find that for any $r\ge e^{4.4}$
	\beqa   \pr\left(e(I,J) \ge \frac rn |I| \cdot |J|Tn\ga\right)
	& \le & \exp\left(-\frac 1n Tn\ga|I|\cdot |J|\left[\frac 12 r\log r - r^{1/2} - r\right]\right) \notag \\
	& \le & \exp\left(-\frac 14 |I|\cdot |J| T\ga r\log r\right). \label{eIJbd}\eeqa
	For a given number  $c$ we define $c_1:=20+2c$  and $k(I,J)$ to be the
	number 
	\[ \text{such that   }  k(I, J) \log[k(I, J)] = c_1\frac{|J|}{|I|\cdot |J|T\ga} \log\frac{n}{|J|},\]
	and let $r(I, J):=\max\{e^{4.4}, k(I, J)\}$. Since $r \mapsto r\log r$ is an increasing function, \eqref{eIJbd} suggests
	\beqax 
	\pr\left(e(I, J) \ge \frac 1n r(I, J) |I| \cdot |J|Tn\ga\right) 
	&\le& \exp\left(-\frac 14 |I|\cdot |J| T\ga k(I, J)\log k(I, J)\right) \\
	&\le& \exp\left(-\frac{c_1}{4} |J|\log\frac{n}{|J|}\right).\eeqax
	Using union bound and the above inequality if 
	\[ \sA_2 :=\left\{e(I, J) < \frac 1n r(I, J) |I| \cdot |J|Tn\ga \text{ for all } I, J \subset [n] \text{ with } |I|\le |J| \le n/e\right\}, \text{ then } \]
	\beqax
	\pr(\sA_2^c)  & \le & \sum_{\{(i,j): 1\le i\le j\le n/e\}} \sum_{\{(I, J): I, J \subset [n]: |I|=i, |J|=j\}} \exp\left(-\frac{c_1}{4}j\log\frac{n}{j}\right) \\
	& \le & \sum_{\{(i,j): 1\le i\le j\le n/e\}} {n\choose i} {n\choose j} \exp\left(-\frac{c_1}{4}j\log\frac{n}{j}\right) \\
	& \le & \sum_{\{(i,j): 1\le i\le j\le n/e\}} (ne/i)^i (ne/j)^j \exp\left(-\frac{c_1}{4}j\log\frac{n}{j}\right) \\
	& \le & \sum_{\{(i,j): 1\le i\le j\le n/e\}}  \exp\left(-\frac{c_1}{4}j\log\frac{n}{j}+i\log\frac{n}{i}+i+j\log\frac{n}{j}+j\right) \\
	& \le & \sum_{\{(i,j): 1\le i\le j\le n/e\}}  \exp\left(-\frac{c_1-16}{2}j\log\frac{n}{j}\right) \\
	& \le & \sum_{\{(i,j): 1\le i\le j\le n/e\}} n^{-\frac 12 (c_1-16)} \le  n^2 \cdot n^{-\frac 12 (c_1-16)} = n^{-\frac 12 (c_1-20)} =n^{-c}.
	\eeqax
	On the event $\sA_2$, if $I, J \subset [n]$ satisfies $|I| \le |J| \le n/e$, then
	\begin{itemize}
		\item[either] $r(I, J) =e^{4.4}$ in which case $e(I, J) \le e^{4.4} |I|\cdot|J|T\ga$,
		\item[or]  $r(I, J) = k(I,J)$ in which case
		\[\frac{e(I, J)}{|I|\cdot|J|T\ga} \log \frac{e(I, J)}{|I|\cdot|J|T\ga}
		\le k(I, J)\log k(I, J) = c_1 \frac{|J|}{|I|\cdot|J|T\ga} \log\frac{n}{|J|}.\]
	\end{itemize}
	This completes the proof.
\end{proof}

\subsubsection{Proof of Theorem \ref{normbd}} \label{normbd proof}
\begin{proof}[Proof of Theorem \ref{normbd}]
	Given $c\in(0,1)$, let $C_1, C_2, c_1$ be the constants and $\sA_1, \sA_2$ be the events appearing in Lemma \ref{Algo1 A bd} and \ref{max degreebd}. We will take $\sA:=\sA_1 \cap \sA_2 \cap \sA_3$, where $\sA_3$ is defined in \eqref{light}, and $C'=C_2$.
	
	We will write $\bar \vA^{(t)}$ (resp.~$\vA$) to denote $\vA^{(t)} - \E \vA^{(t)}$ (resp.~$\vA - \E\vA$).  
	Clearly $\vA-\vP = \bar\vA-Diag(\vP)$ and $||Diag(\vP)||=\max_{a\in[K]} B_{a,a} \le T\ga $. In order to bound $||\bar\vA||$,  we will use the fact
	(see \eg \cite[Lemma B.1]{lei2014consistency}) that  
	\begin{align*} 
		\text{ if } & S := \{\vx=(x_1, x_2, \ldots, x_{n'}): ||\vx||_2 \le 1, 2\sqrt{n'}x_i \in \Z \;\forall \;i\}, \\
		\text{ then } & ||\vW|| \le 4\sup_{\vx,\vy \in S} |\vx^T\vW\vy| \text{ for any symmetric matrix $\vW\in\R^{n'\times n'}$.}
	\end{align*}
	Our argument for bounding $\sup_{\vx,\vy\in S}|\vx^T\bar\vA\vy|$ involves the following two main steps:  bounding the contribution of  (1) {\it light pais} and (2){\it heavy pairs}.
	For $\vx, \vy \in S$, we split the pairs $(x_i, y_j)$ into {\it light pairs} $L$ and {\it heavy pairs} $\bar L$:
	\[ L(\vx, \vy) := \left\{(i,j): |x_iy_j| \le \sqrt{Tn\ga}/n\right\}, \bar L(x, y) := [n']\times[n'] \setminus L(\vx,\vy). \] 
	\begin{enumerate}
		\item {\bf Bounding the contribution of  light pairs.} Here we will show that 
		\begin{equation} \label{light}
			\text{if } \sA_3:=\left\{\sup_{\vx, \vy \in S}\left|\sum_{(i,j) \in L(\vx,\vy)} x_iy_j \bar A_{i,j}\right| \le C_3\sqrt{Tn\ga}\right\}, 
		\end{equation}
		then $\pr(\sA_3)\ge 1- e^{-cn}$, provided
		$C_3>0$ is large enough.
		\item {\bf Bounding  the contribution of heavy pairs.} Here we will show that  there is a constant $C_4>0$ such that
		\begin{equation} \label{heavy}
			\sup_{\vx, \vy \in S}\left|\sum_{(i,j) \in \bar L(\vx,\vy)} x_iy_j\bar A_{i,j}\right| \le C_4 (Tn\ga)^{3/4}
			\text{ on the event } \sA_1\cap\sA_2.
		\end{equation}
	\end{enumerate} 
	This will complete the proof of the theorem.
	
	To show \eqref{light}, we will use the Bernstein's inequality. Fix $\vx, \vy \in S$. Define
	$u_{i,j} := x_iy_j \mathbf 1(|x_iy_j| \le \sqrt{T\ga/n}+ x_jy_i \mathbf 1(|x_jy_i| \le \sqrt{T\ga/n}$ for $i, j \in [n']$.
	Clearly $|u_{i,j}| \le 2\sqrt{T\ga/n}$ and
	\beq \label{ubd1}
	\sum_{1\le i<j\le n'} u_{i,j}^2 \le 2 \sum_{1\le i<j\le n'} [(x_iy_j)^2 + (x_jy_i)^2 ] = 2||\vx||_2^2||\vy||_2^2\le2 \eeq
	It is easy to see that $\sum_{(i,j)\in L(\vx,\vy)}x_iy_j\bar A_{i,j} = \sum_{1\le i<j\le n', t\in[T]} u_{i,j} \bar A^{(t)}_{ij}$. Also,  the summands in the last sum are independent, each summand has mean 0 and is bounded by $2\sqrt{T\ga/n}$.  So, using union bound, Bernstein's inequality~and~\eqref{ubd1},
	\begin{align*} 
		\pr(\sA_3^c) \le & |S|^2 \exp\left(-\frac{\frac 12 C_3^2Tn\ga}{\sum_{1\le i<j\le n', t\in[T]}u_{i,j}^2 P^{(t)}_{i,j} + \frac 23\sqrt{T\ga/n}\cdot C_3\sqrt{Tn\ga}}\right) \\
		\le & |S|^2 \exp\left(-\frac{\frac 12 C_3^2Tn\ga}{2T\ga + \frac {2C_3T\ga}{3} }\right) 
		\le  |S|^2 \exp\left(-n\frac{C_3^2}{4+\frac 43 C_3}\right).
	\end{align*} 
	A standard volume argument (see \eg \cite[Claim 2.9]{FO05}) gives $|S| \le e^{\log(14) n}$. This together with the last display proves \eqref{light}.
	
	To show \eqref{heavy}, first note that for any $\vx, \vy \in S$
	\begin{align*}
		\left|\sum_{(i,j) \in \bar L(\vx,\vy)} x_i y_j P_{i,j}\right|
		\le \sum_{(i,j) \in \bar L(\vx,\vy)} \frac{x_i^2 y_j^2}{|x_i y_j|}  
		T\ga
		\le \sqrt{Tn\ga} ||\vx||_2^2 ||\vy||_2^2 
		\le \sqrt{Tn\ga}. 
	\end{align*}
	So it remains to bound 
	$\sup_{\vx, \vy \in S} |\sum_{(i,j) \in \bar L(\vx,\vy)} x_iy_j A_{i,j}|$.
	We also note that each $\bar L$ is union of four sets $\bar L_{++}, \bar L_{+-}, \bar L_{-+}$ and $\bar L_{--}$, where 
	$\bar L_{\pm\pm}(\vx, \vy) :=  \{(i,j) \in \bar L(\vx,\vy): \pm x_i> 0, \pm y_j> 0\}$, and  it suffices to bound
	$\sup_{\vx, \vy \in S} |\sum_{(i,j) \in \bar L_{++}(\vx,\vy)} x_iy_j A_{i,j}|$,
	as the arguments for bounding the other three terms are similar.
	To do so we fix $\vx, \vy \in S$ and define the index sets
	\begin{align*}
		I_1:=\left\{i:\frac{2^{-1}}{\sqrt n} \le x_i \le \frac{1}{\sqrt n}\right\}, I_s:=\left\{i:\frac{2^{s-2}}{\sqrt n} < x_i \le \frac{2^{s-1}}{\sqrt n}\right\},  \\
		J_1:=\left\{j:\frac{2^{-1}}{\sqrt n} \le y_j \le \frac{1}{\sqrt n}\right\}, J_t:=\left\{j:\frac{2^{t-2}}{\sqrt n} < y_j \le \frac{2^{t-1}}{\sqrt n}\right\}  
	\end{align*}
	for $s, t = 2, 3,  \ldots\lceil\log_2(2\sqrt n)\rceil$.
	It is easy to see that $|I_s|\le 2^{s}$ and $|J_t| \le  2^t$, both of which are at most $4\sqrt n$.
	
	For two subsets of vertices $I, J\subset [n']$, let 
	\begin{align*}
		& e(I,J):=\sum_{i\in I, j\in J} A_{i,j},  \lambda_{s,t}:=\frac{e(I_s,J_t)}{|I_s|\cdot |J_t|T\ga}, \alpha_s:=|I_s|2^{2s}/n, \beta_t:=|J_t|2^{2t}/n, \text{ so}\\
		& \sum_{(i,j) \in \bar L_{++}(\vx,\vy)} x_iy_j A_{i,j}\le \sum_{(s,t): 2^{s+t} \ge \sqrt{Tn\ga}} e(I_s, J_t)\frac{2^{s+t}}{n} 
		\le  \sqrt{Tn\ga} (L_+ + L_-), \text{ where }\\
		& L_\pm:=\sum_{(s,t)\in\sS_\pm} \alpha_s\beta_t \frac{\lambda_{s,t} \sqrt{Tn\ga}}{2^{s+t}},  \sS_\pm:= \{(s, t) : 2^{s+t} \ge \sqrt{Tn\ga}, \pm(|I_s| -  |J_t|) \le 0\}.
	\end{align*}
	Now note that the argument for bounding $L_+$ can be immited (after interchanging the role of $(s,I_s,\ga_s)$ and $(t, J_t,\gb_t)$) to give a similar bound for $L_-$. So,  it suffices to show
	\beq \label{sumbd}
	L_+ \le C_5 (Tn\ga)^{1/4} \text{ on the event } \sA_1 \cap \sA_2.
	\eeq
	In order to show  \eqref{sumbd}, we further divide $\sS_+$ into subsets
	\beqax
	\sS_1 & := & \{(s, t) \in \sS_+ : \gl_{s,t} \le 2^{s+t}/\sqrt{Tn\ga}\} \\
	\sS_2 & := & \{(s, t) \in \sS_+\setminus \sS_1 : \gl_{s,t} \le e^{4.4}\} \\
	\sS_3 & := & \{(s, t) \in \sS_+\setminus \cup_{i=1}^2 \sS_i : 2^s \ge 2^t\sqrt{Tn\ga}\} \\
	\sS_4 & := & \{(s, t) \in \sS_+\setminus \cup_{i=1}^3 \sS_i: \log \gl_{s,t} \ge \frac 14 [2t \log 2 + \log(1/\gb_t)]\} \\
	\sS_5 & := & \{(s, t) \in \sS_+\setminus \cup_{i=1}^4 \sS_i : 2t \log 2 \ge \log(1/\gb_t)\} \\
	\sS_6 & := & \sS_+\setminus \cup_{i=1}^5\sS_i,
	\eeqax
	and show how to bound the sums $\sum_{(s,t)\in\sS_i} \ga_s\gb_t\gl_{s,t}\sqrt{Tn\ga}2^{-s-t}, i\in[6],$ separately. We will need the fact
	$\sum_s \ga_s \le \sum_i 4x_i^2 \le 4$ and  $\sum_t \gb_t \le \sum_i 4y_i^2 \le 4$.
	Using the last bound
	\beqa \label{S12bd} 
	&& \sum_{(s, t) \in \sS_1}  \alpha_s\beta_t \frac{\lambda_{s,t} \sqrt{Tn\ga}}{2^{s+t}}
	\le \sum_{s,t} \ga_s\gb_t \le 16, \\
	&&  \sum_{(s, t) \in \sS_2}  \alpha_s\beta_t \frac{\lambda_{s,t} \sqrt{Tn\ga}}{2^{s+t}}
	\le \sum_{(s, t) \in \sS_2}  \alpha_s\beta_t \lambda_{s,t} 
	\le e^{4.4} \sum_{s,t} \ga_s\gb_t \le 16e^{4.4}. \notag\eeqa
	Next note that $e(I,J) \le |I| \max_{i\in [n']} \sum_{j\in[n'], t\in [T]}A^{(t)}_{i,j}\le |I|(T\bar d)^{5/4}$ on the event $\sA_1$.
	Also note that for any $s$ the sum $\sum_{t: (s,t) \in \sS_3} \sqrt{ Tn\ga }2^{-s+t} \le \sum_{i\ge 0}2^{-i}$ by the definition of $\sS_3$. Using these bounds and the definition of $\sA_1$,
	\begin{align} 
		\sum_{(s, t) \in \sS_3}  \alpha_s\beta_t \frac{\lambda_{s,t} \sqrt{ Tn\ga }}{2^{s+t}}
		\le  \sum_{s} \ga_s \sum_{t: (s,t) \in \sS_3} \frac{|J_t|2^{2t}}{n}  \frac{ |I_s|(T\bar d)^{5/4}}{|I_s|\cdot |J_t| T\ga}\sqrt{Tn\ga}2^{-s-t} \notag \\
		\le  \frac{ (T\bar d)^{5/4}}{Tn\ga} \sum_{s} \ga_s \sum_{t: (s,t) \in \sS_3} \frac{\sqrt{Tn\ga}}{2^{s-t}}  
		\le 8\frac{ (T\bar d)^{5/4}}{Tn\ga} = 8C_1^{5/4} (Tn\ga)^{1/4} \text{ on } \sA_1.  \label{S3bd}
	\end{align}
	On the event $\sA_2$, it is easily seen that for each $(s,t) \in \cup_{i=4}^6\sS_i$ (as $(s,t) \not\in \sS_2$)
	\[ e(I_s, J_t) \log\frac{e(I_s, J_t)}{|I_s|\cdot |J_t| T\ga} \le c_1  |J_t|\log\frac{n}{|J_t|}\]
	which can be checked (after a straight forward algebraic manipulation)  to be  equivalent with the condition
	\beq\label{eIJcond}
	\ga_s\frac{\sqrt{ Tn\ga }}{2^{s+t}} \gl_{s,t}\log\gl_{s,t} 
	\le c_1 \frac{2^{s-t}}{\sqrt{ Tn\ga }}
	[2t\log 2+\log(\gb_t^{-1})].
	\eeq
	Now, if $(s,t) \in \sS_4$, then \eqref{eIJcond} will imply $\ga_s\gl_{s,t}\sqrt{ Tn\ga }2^{-s-t} \le 4c_1 2^{s-t}/\sqrt{ Tn\ga }$ Also, for any $t$  the sum $\sum_{s: (s,t) \in \sS_4} 2^{s-t}/\sqrt{ Tn\ga }$ is at most the geometric sum $\sum_{i\ge 0} 2^{-i}$, because $(s,t) \in \sS_4$ ensures (as $(s,t) \not\in \sS_3$) $2^{s-t} \le \sqrt{ Tn\ga }$. From these two observations it follows easily that
	\beq
	\sum_{(s,t) \in \sS_4}  \frac{\ga_s\gb_t \gl_{s,t}\sqrt{ Tn\ga }}{2^{s+t}}
	\le \sum_t \gb_t \sum_{s: (s,t) \in \sS_4}   \frac{4c_12^{s-t}}{\sqrt{ Tn\ga }} 
	\le 8c_1\sum_t \gb_t \le 32c_1. \label{S4bd}
	\eeq
	Next we see that if $(s,t) \in \sS_5$, then (as $(s,t) \not\in \sS_2$) $\log\gl_{s,t} \ge 1$ which gives a lower bound for the LHS of \eqref{eIJcond}. We also get an upper bound for the RHS of \eqref{eIJcond} replacing $\log(\gb_t^{-1})$ by another $2t\log 2$. Combining these two bounds,  
	$\ga_s\gl_{s,t} \sqrt{ Tn\ga }2^{-s-t} \le c_1\frac{2^{s-t}}{\sqrt{ Tn\ga }} 4t\log 2\le 4\log(2)c_12^s/\sqrt{ Tn\ga }$
	for all $(s,t) \in \sS_5$. On the other hand, $(s,t) \in \sS_5$ implies
	$\log\gl_{s,t} \le \frac 14[2t\log 2 + \log(\gb_t^{-1})] \le t\log 2$.  Also, $(s,t) \in \sS_5$ ensures (as $(s,t) \not\in \sS_1$) $\gl_{s,t}\sqrt{ Tn\ga }2^{-s-t} \ge 1$. Combining these two facts we get $2^s \le \sqrt{ Tn\ga }$, which implies $\sum_{s:(s,t) \in \sS_5} 2^s/\sqrt{ Tn\ga } \le \sum_{i\ge 0} 2^{-i}$ for each $t$.  Therefore, 
	\beq
	\sum_{(s,t) \in \sS_5}  \frac{\ga_s\gb_t \gl_{s,t}\sqrt{ Tn\ga }}{2^{s+t}}
	\le 4\log(2)c_1\sum_t \gb_t \sum_{s: (s,t) \in \sS_5} \frac{2^s}{\sqrt{ Tn\ga }} 
	\le 32c_1
	\label{S5bd}\eeq
	Finally for $(s,t) \in \sS_6$  , $\log\gl_{s,t}$ is at most $\frac 14 [2t\log 2 + \log(\gb_t^{-1})]$ (as $(s,t) \not\in\sS_4$), which is bounded by $\frac 12\log(\gb_t^{-1})$. $(s,t) \in \sS_6$ also ensures (as $(s,t) \not \in \sS_2$) $\log\gl_{s,t}$ is positive. Combining these two facts,  $\gb_t\gl_{s,t} \le 1$ for $(s,t) \in \sS_6$. Combining this with the fact $\sum_{t: (s,t) \in \sS_+} \sqrt{ Tn\ga } 2^{-s-t} \le \sum_{i\ge 0} 2^{-i}$ we get
	\beq
	\sum_{(s,t) \in \sS_6} \ga_s\gb_t \gl_{s,t} \frac{\sqrt{ Tn\ga }}{2^{s+t}}
	\le \sum_s \ga_s \sum_{t: (s,t) \in \sS_6} \sqrt{ Tn\ga } 2^{-s-t} \le 2\sum_s \ga_s \le 8.
	\label{S6bd}\eeq
	Combining the conclusions of \eqref{S12bd},  \eqref{S3bd},  \eqref{S4bd},  \eqref{S5bd} and ,  \eqref{S6bd} completes the argument for showing \eqref{sumbd}, and hence proves the desired theorem.
\end{proof}

\subsection{Consistency of Spherical Spectral Clustering Labels $\check\vZ_0$ under DDCBM}
In this section, we prove the result about the consistency of $\check\vZ_0$ for networks generated from \textbf{DDCBM}.
\begin{thm} \label{ConsSum2}
	Let $\{\vA^{(t)}\}_{t=1}^T$ be the adjacency matrices of the networks generated from the DDCBM
	with parameters $(\vZ, \V{\pi}, \{\vB^{(t)}\}_{t=1}^T, \V{\psi})$, where
	\begin{itemize}
		\item $T\ge 1$ is the number of networks
		\item $n$ is the number of nodes, $K$ is the number of communities (having labels $\cC_1, \cC_2, \ldots, \cC_K$) obeying $n\ge 3K$
		\item $\{\vB^{(t)}\}_{t=1}^T$ satisfy assumption  \eqref{eq_sum_ass}
		\item $\V{\psi}$  satisfies \eqref{eq_dcbm_id}
		\item $\tau_k := \sum_{i\in\cC_k} \psi_i^2 \sum_{i\in\cC_k} \psi_i^{-2}, k\in[K],$ be a measure of heterogeneity of $\V{\psi}$.
		\item $\ga=\ga(n,T):=\max_{a,b \in[K], t\in [T]} B^{(t)}_{a,b}$ is the maximum connection probability of an edge at any time, and $\gl=\gl(n,T)>0$ is such that $\gl\ga$ is the smallest eigenvalue of  $(\vB^{(t)}, t\in[T])$.
	\end{itemize}
	For any $\gee>0$ and $c\in(0,1)$, 
	there are constants $C_1=C_1(\gee,c), C_2=C_2(c)>0$ such that if $Tn\ga\ge C_2(K/\gl)^5$
	and if  $\check\vZ_0$ is the estimate of $\vZ$ as described in Algorithm 2, then the overall fraction of misclassified nodes in $\check\vZ_0$ is 
	\beq \label{overall misclassify} 
	\le e^{-(1-c)Tn\ga}+C \frac{\left(\sum_{k\in[K]} \tau_k\right)^{1/2}}{\min_{k\in[K]}\sum_{i\in\cC_k\cap\{k_1, k_2, \ldots, k_{n'}\}}\psi_i^2} \frac{\sqrt K}{\gl (Tn\ga)^{1/4}}. \eeq
	with probability at least $1-5\exp(-c\min\{Tn\ga\gl, \log(n)\})$.
	
	Therefore, in the special case,  when 
	(i)  $\gl>0$, \; (ii) the community sizes are balanced, i.e.~$n_{\text{max}}/n_{\text{min}} = O(1)$ and
	(iii) $\psi_i=\alpha_i/\max\{\alpha_j: z_i=z_j\}$, where $(\alpha_i)_{i=1}^n$ are  i.i.d.~weights having bounded support, 
	then consistency holds for $\check\vZ_0$ with probability $1-o(1)$ if $\E\alpha_1^{-2}<\infty$ and $Tn\ga\gl\to\infty$.
\end{thm}

\subsubsection{Proof of Theorem \ref{ConsSum2}}
\label{sec_proof_checkZ_supp}
\begin{proof}[Proof of Theorem \ref{ConsSum2}]
	Without loss of generality, we can assume that $k_i=i$ and $l_j=j$ for all $i\in[n']$ and $j\in[n'']$.
	Define 
	\begin{align} 
		\V{\Psi} := Diag(\V{\psi})\cdot\vZ,  \quad \tilde n_k':=\norm{\V{\Psi}_{[n'],k}}_2^2 \text{ for } k\in [K], \notag \\
		\tilde\GD := Diag\left(\sqrt{\tilde n_1'}, \sqrt{\tilde n_2'}, \ldots, \sqrt{\tilde n_K'}\right),  \quad \tilde \vB  := \sum_{t=1}^T \tilde \GD \vB^{(t)}\tilde \GD,\notag \\
		\tilde \vP  := \sum_{t=1}^T \V{\Psi}_{[n'],*} \vB^{(t)}\V{\Psi}_{[n'],*}^T = \V{\Psi}_{[n'],*}\tilde \GD^{-1}\tilde \vB\tilde \GD^{-1}\V{\Psi}_{[n'],*}^T. 
		\label{Ptildef}
	\end{align}
	Using the fact that the columns of $\V{\Psi}_{[n'],*}\tilde\GD^{-1}$ are orthonormal, it is clear from the last display that if the eigenvalue decomposition  of $\tilde\vB$ is $\tilde\vR\tilde\vD\tilde\vR^T$, then the leading eigenvectors of $\tilde\vP$ are given by the columns of $\V{\Psi}_{[n'],*}\tilde\GD^{-1}\tilde\vR =: \tilde\vU$.  
	Following the argument which leads to \eqref{Frob_bd} we get that there is an orthogonal $\tilde\vQ \in \R^{K \times K}$ such that 
	\beq  \label{Frob_bd4}
	\norm{\hat \vU - \tilde\vU\tilde\vQ}_F \le \frac{2\sqrt{2K}}{\gamma_n}\norm{\vA - \tilde\vP}.
	\eeq
	Also, one can  use the fact that $\tilde P^{(t)}_{i,j} \le P^{(t)}_{i,j}$ for all $i, j, t$ and repeat the  argument leading to \eqref{sA^c bd} with $\vP$ replaced by $\tilde\vP$, to conclude that for any $c\in(0,1)$ there are constants $C, C'>0$ (depending on $c$) such that if $Tn\ga \ge C'(K/\gl)^5$~and
	\beq  \label{spectralnormbd4}
	\tilde\sA := \left\{\norm{\vA - \tilde\vP} \le C(Tn\ga)^{3/4} \text{ and } n'\ge n(1-e^{-(1-c)Tn\ga})\right\},  
	\eeq
	then $\pr(\tilde\sA)\ge 1-5\exp[-c\min\{\log(n), Tn\ga\gl\}]$.

	Next, we normalize the nonzero rows of $(\tilde\vU\tilde\vQ)_{[n''],*}$ to obtain  
	$\tilde\vU^+$,  and 
	\[   \text{define  }\check S  := \left\{j \in [n'']: \norm{\check\vZ^+_{j,*}\check \vX - \tilde\vU^+_{j,*}}_2 \ge 1/\sqrt 2\right\} \quad \text{ and } \quad  \check S' := [n''] \setminus \check S.\]
	Using standard argument (see Appendix for details), 
	\begin{lem} \label{CorrectClassify}
		All nodes with indices in $\check S'$ are correctly classified in $\check\vZ$, and $n'-|\check S'|\le  C(\eps)  \left(\sum_{k\in[K]} \tau_k\right)^{1/2} \norm{\hat  \vU-\tilde\vU\tilde\vQ}_F$ for some constant $C(\eps)$. 
	\end{lem}
	Lemma \ref{CorrectClassify} together with \eqref{Frob_bd4} and \eqref{spectralnormbd4} implies that, on the event $\tilde\sA$,  the overall fraction of misclassified nodes   is at most
	\beq \label{total misclassify} 
	(n - |\check S'|)/n \le  e^{-(1-c)Tn\ga} + C\frac 1n\left(\sum_{k\in[K]} \tau_k\right)^{1/2}  \frac{\sqrt K}{\gc_n} (Tn\ga)^{3/4}. \eeq
	In order to estimate $\gc_n$, note that
	$\cC(\V{\Psi})^\perp \subset \cN(\tilde\vP)$. This together with the fact that $\V{\Psi}^T\V{\Psi}=\tilde\GD^2$ implies
	\beqax 
	\gamma_n & \ge & \min_{\vz\ne\mathbf 0: \vz\in \cC(\V{\Psi})} \frac{\vz^T\tilde\vP\vz}{\vz^T\vz} \ge\min_{\vx\in \R^K\setminus \{\mathbf 0\}} \frac{\vx^T\V{\Psi}^T\tilde\vP\V{\Psi} \vx}{\vx^T\V{\Psi}^T\V{\Psi} \vx}  \\
	& \ge &  \min_{\vy\in \R^K\setminus \{\mathbf 0\}} \frac{\vy^T\sum_{t\in[T]}\tilde\GD \vB^{(t)}\tilde\GD\vy}{\vy^T\vy} 
	\ge \sum_{t\in[T]} \min_{\vy\in \R^K\setminus \{\mathbf 0\}} \frac{\vy^T\tilde\GD\vB^{(t)}\tilde\GD\vy}{\vy^T\vy}.\eeqax
	Here we have used the change of variable $\vz=\V{\Psi}\vx$ and $\vy=\tilde\GD \vx$. Bounding each of the summands in the above lower bound by the corresponding smallest eigenvalue and using the definition of $\gl$   we get 
	\[ \gc_n \ge T\ga \min_{t\in[T]} \gl_1\left((\tilde\GD \frac 1\ga \vB^{(t)}\tilde\GD)\right)  
	\ge T \ga \tilde n'_{\text{min}}\gl.\]
	Plugging this bound for $\gamma_n$ in \eqref{total misclassify} we prove the estimate in \eqref{overall misclassify}. 
	
	If $(\psi_i)$ satisfy conditions (i) - (iii), then the fraction of misclassified nodes in \eqref{overall misclassify} is $o(1)$ with probability $1-o(1)$, provided $ Tn\ga \gl\to\infty$ and $\E\alpha_1^{-2}<\infty$. This completes the argument.
\end{proof}

\section{Conclusion and Future Works}
In this paper, we consider the dynamic stochastic block model with constant community memberships and changing connectivity matrices. We consider spectral clustering and spherical spectral clustering algorithms on aggregate versions of adjacency matrices, based on the sum of adjacency matrices. It is shown in the paper that under dynamic stochastic block model, spectral clustering based on the sum of squared adjacency matrices has guarantee of community recovery, under associative community structure. We also consider spherical spectral clustering based on the sum of adjacency matrices and give theoretical guarantee that the spherical spectral clustering method recovers associative community membership under dynamic degree-corrected block model. 

\subsection{Future Works}
Several extensions are possible from the current work. Some possible extensions of our work will include considering the cases when cluster memberships change with time and the dynamic behavior of the networks are more general, such as, dependence of adjacency matrices on edge structure and community memberships of previous time points. Methods for community recovery with theoretical guarantee are quite rare in the literature and it would be good to investigate such problems in later works.

\section*{Acknowledgements}
We thank Peter Bickel, Paul Bourgade, Ofer Zeitouni and Harrison Zhou for helpful discussions and comments.


\section*{Appendix}
\subsection{Proof of Lemma \ref{misclassify bd}} \label{misclassify bd proof}
\begin{proof}[Proof of Lemma \ref{misclassify bd}] Note that \eqref{eq_sum_ass} implies $rank(\vP)=K$, so applying \cite[Lemma 5.1]{lei2014consistency}  there is an orthogonal matrix $\vQ \in \R^{K \times K}$ such that 
	\begin{equation} \label{Frob_bd}
		||\hat \vU - \vU\vQ||_F \le \frac{2\sqrt{2K}}{\gamma_n}||\vA-\vP||.
	\end{equation}
	Letting $\vV:=\GD^{-1}\vR\vQ$ and noting that $\vV\vV^T=\GD^{-2}$, we see that $\vU\vQ=\vZ_{[n'],*} \vV$, where the rows of $\vV$ are orthogonal and $||\vV_{i,*}||_2=1/\sqrt{n_i'}$.   For $i \in [K]$ define $T_i := \{j \in [n']: \vZ_{j,i}=1,  ||\hat\vZ_{j,*}\hat\vX - \vZ_{j,*}\vV||_2\le ||\vV_{i,*}||/2\}$ and
	\[
	S_i := \{j \in [n']: \vZ_{j,i}=1, j\notin T_i\}, 
	\text{ so } ||\hat\vZ\hat X -\vZ_{[n'],*} \vV||_F^2 \ge \frac 14 \sum_{i=1}^K |S_i|\;||\vV_{i,*}||^2.
	\]
	Now using triangle inequality and \eqref{eq:kmeans}
	\begin{align*}
		& ||\hat\vZ\hat\vX - \vZ_{[n'],*}\vV||_F^2
		\le 2(||\hat\vZ\hat\vX - \hat\vU||_F^2 + || \vZ_{[n'],*}\vV-\hat\vU||_F^2)
		\le 2(2+\eps)||\vZ_{[n'],*}\vV-\hat\vU||_F^2,
	\end{align*}
	as $\vZ_{[n'],*}\in\sM_{n',K}$.  Since $\vZ_{[n'],*}\vV=\vU\vQ$, the~last~two~displays~and~\eqref{Frob_bd}~give
	\[
	\sum_{i=1}^K |S_i|\;||\vV_{i,*}||^2 \le 32K\gc_n^{-2}(4+2\eps) ||\vA-\vP||^2.
	\]
	Now note that
	whenever $j_1 \in T_{i_1}$ and $j_2 \in T_{j_2}$ for some $j_1, j_2 \in [n']$ and $i_1, i_2 \in [K]$ with $i_1 \ne i_2$,
	one must have $(\hat\vZ\hat X)_{j_1,*} \ne (\hat\vZ\hat X)_{j_2,*}$,  since otherwise
	one can use triangle inequality and the fact that $(\vZ\vV)_{j_l,*}=\vV_{i_l,*}, l=1, 2,$
	\beqax  \text{ to have}
	&&||\vV_{i_1,*}||^2 + ||\vV_{i_2,*}||^2 = ||\vV_{i_1,*} - \vV_{i_2,*}||^2  \le 2(||(\hat\vZ\hat\vX)_{j_1,*} - (\vZ\vV)_{j_1,*}|| \\
	&& \quad + ||(\hat\vZ\hat\vX)_{j_2,*} - (\vZ\vV)_{j_2,*}||)
	\le \frac 12(||\vV_{i_1,*}||^2 + ||\vV_{i_2,*}||^2),\eeqax
	which gives a contradiction.
	Also, whenever $j_1, j_2 \in T_i$ for some $j_1, j_2 \in [n']$ and $i \in [K]$,
	one must have $(\hat\vZ\hat X)_{j_1,*} = (\hat\vZ\hat X)_{j_2,*}$, since otherwise $\hat\vZ\notin \sM_{n',K}$. 
	
	Thus all nodes outside $\cup_{i=1}^K S_i$ are correctly classified.
	This completes the proof, as $||\vV_{i,*}||^2\ge 1/n_i$ for all $i\in[K]$.
\end{proof}

\subsection{Proof of Lemma \ref{CorrectClassify}}
\begin{proof}[Proof of Lemma \ref{CorrectClassify}]
	Recall the definition of $\tilde\vU, \V{\Psi}, \tilde\GD, \tilde \vR, \tilde \vQ$, and let $\tilde\vV:=\tilde\GD^{-1}\tilde\vR\tilde\vQ$.  Then, $\tilde\vV\tilde\vV^T=\tilde\GD^{-2}$, as $\tilde\vR$ and $\tilde\vQ$ are orthogonal, and $\tilde\vU\tilde\vQ=\V{\Psi}_{[n'],*}\tilde\vV$. This implies  
	$\norm{(\tilde\vU\tilde\vQ)_{i,*}}_2=\psi_i/\sqrt{\tilde n'_{z_i}}$. So, if $\tilde\vU^+$ consists of  normalized rows of $(\tilde\vU\tilde\vQ)_{[n''],*}$, then 
	\beqa \label{U^+def} 
	\tilde\vU^+ 
	&:=& Diag\left(\sqrt{\tilde n_{z_1}}, \ldots, \sqrt{\tilde n_{z_{n''}}}\right)  [Diag(\V{\psi})]^{-1}_{[n''],[n'']}\V{\Psi}_{[n''],*}\tilde\vV \notag \\
	&=& Diag\left(\sqrt{\tilde n_{z_1}}, \ldots, \sqrt{\tilde n_{z_{n''}}}\right) \vZ_{[n''],*} \tilde\vV=\vZ_{[n''],*} \tilde\GD\tilde\vV, \text{ and }\\
	\tilde\vU^+(\tilde\vU^+)^T &=& (\vZ\vZ^T)_{[n''],[n'']}  \text{ as } \tilde\vV\tilde\vV^T=\tilde\GD^{-2}.\notag
	\eeqa 
	Now recall the definition of $\check S$ and $\check S'$. 
	If $k, l \in \check S'$  satisfy $\vZ_{k,*} \ne \vZ_{l,*}$, then the fact $\tilde\vU^+(\tilde\vU^+)^T=(\vZ\vZ^T)_{[n''],[n'']}$ and triangle inequality  give
	\begin{align*}
		\sqrt 2
		& =  \norm{\tilde\vU^+_{k,*}  - \tilde\vU^+_{l,*}}_2  \le  \norm{\tilde\vU^+_{k,*}  - \check\vZ_{k,*} \check\vX}_2 + \norm{\tilde\vU^+_{l,*}  - \check\vZ_{l,*} \check\vX}_2 + \norm{\check\vZ_{k,*} \check\vX - \check\vZ_{l,*} \check\vX}_2  \\ 
		&<  \left(\frac{1}{\sqrt{2}} + \frac{1}{\sqrt{2}}\right) + \norm{\check\vZ_{k,*} \check \vX - \check\vZ_{l,*} \check \vX}_2 \le \sqrt 2  + \norm{\check\vZ_{k,*} \check \vX - \check\vZ_{l,*} \check \vX}_2.
	\end{align*}
	The above shows  $\check\vZ_{k,*} \ne \check\vZ_{l,*}$.
	On the other hand,  if $k, l \in \check S'$ are such that $\vZ_{k,*}=\vZ_{l,*}$,
	then $\check\vZ_{k,*} = \check\vZ_{l,*}$ must hold, because otherwise $\check\vZ^+\not\in\sM_{n'', K}$.
	In other words, all nodes with indices in $\check S'$ are correctly classified.
	In order to estimate $n'-|\check S'|=n'-n''+|\check S|$  note that 
	\begin{align*}
		& \norm{\check \vZ^+\check  \vX -\tilde\vU^+}_{2,1} \ge  \sum_{j\in \check S} \norm{\check \vZ^+_{j,*}\check  \vX -\tilde\vU^+_{j,*}}_2 \ge  \frac{|\check S|}{\sqrt 2} \text{ using the definition of $\check S$,} \\
		& \norm{\check \vZ^+\check  \vX -\tilde\vU^+}_{2,1} \le \norm{\check \vZ^+\check  \vX -\hat  \vU^+}_{2,1}+\norm{\hat  \vU^+ -\tilde\vU^+}_{2,1}\text{ using triangle inequality, and} \\
		& \norm{\check \vZ^+\check  \vX -\hat  \vU^+}_{2,1} \le (1+\eps) \min_{\V{\gC}\in\sM_{n'',K}, \vX\in\R^{K\times K}} \norm{\V{\gC} \vX -\hat  \vU^+}_{2,1} \le  (1+\eps) \norm{\tilde\vU^+ -\hat  \vU^+}_{2,1} \end{align*}
	using the definition of $(\check \vZ^+, \check  \vX)$  and noting that $\tilde\vU^+=\vZ_{[n''],*}\tilde\GD\tilde\vV$ (see \eqref{U^+def}) is a candidate for the above minimization problem.  The last three inequalities give  
	\[
	|\check S| \le \sqrt 2(2+\eps) \norm{\hat  \vU^+-\tilde\vU^+}_{2,1}.\]
	Now using a standard inequality that $||(\vx/||\vx||) - (\vy/||\vy||)|| \le 2||\vx-\vy||/||\vx||$ for any $\vx, \vy \ne \vzero$ (see e.g.~\cite[page 16]{lei2014consistency}), the Cauchy-Schartz inequality and the fact that $\tilde\vU\tilde\vQ=\V{\Psi}_{[n'],*}\tilde\vV$ we get
	\begin{align*}
		\norm{\hat  \vU^+-\tilde\vU^+}_{2,1} 
		& \le  2 \sum_{i\in[n']} \frac{\norm{\hat  \vU_{i,*} - (\tilde\vU\tilde\vQ)_{i,*}}_2}{\norm{(\tilde\vU\tilde\vQ)_{i,*}}_2}
		\le \norm{\hat  \vU - \tilde\vU\tilde\vQ}_F \left(\sum_{i\in[n']} \norm{(\V{\Psi}\tilde\vV)_{i,*}}_2^{-2}\right)^{1/2} \\
		(n'-n'')^2 & = \left(\sum_{i\in[n']} \mathbf 1(\hat  \vU_{i,*}=0)\norm{(\tilde\vU\tilde\vQ)_{i,*}}_2
		\norm{(\tilde\vU\tilde\vQ)_{i,*}}_2^{-1}\right)^2 \\ 
		& \le  \sum_{i\in[n']} \mathbf 1(\hat  \vU_{i,*}=0)\norm{(\tilde\vU\tilde\vQ)_{i,*}}_2^2
		\sum_{i\in[n']}\norm{(\tilde\vU\tilde\vQ)_{i,*}}_2^{-2}
		\le \norm{\hat  \vU-\tilde\vU\tilde\vQ}_F^2  \sum_{i\in[n']} \norm{(\V{\Psi}\tilde\vV)_{i,*}}_2^{-2}
	\end{align*}
	Now note that $||\tilde\vV_{k,*}||_2=1/\sqrt{\tilde n'_k}$ for each $k\in[K]$. So, after rearranging its summands,  the sum 
	\[ \sum_{i\in[n']} \norm{(\V{\Psi}\tilde\vV)_{i,*}}_2^{-2}  \le\sum_{k\in[K]} \sum_{i\in\cC_k} \psi_i^{-2}\tilde n'_k \le \sum_{k\in[K]} \left[\sum_{i\in \cC_k} \psi_i^2 \sum_{i\in \cC_k} \psi_i^{-2}\right] =\sum_{k\in[K]} \tau_k.\]
	The last four displays give
	$n'-|\check S'| \le C(\eps)  \left(\sum_{k\in[K]} \tau_k\right)^{1/2} \norm{\hat  \vU-\tilde\vU\tilde\vQ}_F$. 
\end{proof}

\bibliographystyle{imsart-nameyear}
\bibliography{Biom}

\end{document}